\documentclass{article}
\usepackage[T1]{fontenc}
\usepackage{amsthm}
\usepackage{syntax, amsmath, amssymb, stmaryrd, mathtools, semantic, enumitem, booktabs, listings, textcomp, placeins, algorithm, algorithmicx, algpseudocode, subcaption, xr, multicol, float, array, authblk, xspace, graphicx, hyperref, cite, thmtools}
\usepackage[dvipsnames]{xcolor}
\usepackage{todonotes}

\newtheorem{definition}{Definition}
\newtheorem{lemma}{Lemma}
\newtheorem{corollary}{Corollary}

\newfloat{Syntax}{htp}{syn}
\lstdefinelanguage{ADL}{
	keywords = [1]{struct, int, bool, nat, String, if, then, Fix, else},
	morekeywords=[2]{null, this, true, false},
	sensitive=false, 
	morecomment=[l]{//},
	morecomment=[s]{/*}{*/},
	morestring=[b]"
}

\lstdefinelanguage{gmcommands}{
	morekeywords = {push, rd, wr, cons, if, not, op, this, Fix},
	sensitive=true, 
	morecomment=[l]{//},
	morecomment=[s]{/*}{*/},
	morestring=[b]"
}

\definecolor{codegray}{rgb}{0.5,0.5,0.5}

\lstdefinestyle{gmcommands}{
	language={gmcommands},
	showstringspaces=false,
	basicstyle=\small,
	keywordstyle=\textbf,
	numberstyle=\tiny\color{codegray},
	stringstyle=\slshape,
	commentstyle=\color{codegray},
	emph={
		val
	},
	emphstyle = \rmfamily\itshape,
	breaklines=true
}

\lstdefinestyle{CUDA}{
	language=[ANSI]C++,
	showstringspaces=false,
	basicstyle=\small,
	keywordstyle=\color{MidnightBlue},
	numberstyle=\tiny\color{codegray},
	stringstyle=\slshape,
	commentstyle=\color{codegray},
	emph={
		cudaMalloc, cudaFree, cudaMemcpy, cudaMemcpyHostToDevice, cudaMemcpyDeviceToHost, cudaDeviceSynchronize,
		__global__, __shared__, __device__, __host__,
		__syncthreads,
	},
	emphstyle = \color{OliveGreen},
	breaklines=true
}

\newcommand{\etal}{\emph{et al.}}
\newcommand{\Lname}{AuDaLa\xspace}
\newcommand{\rarr}{\rightarrow}

\newcommand{\vempty}[0]{\varepsilon}
\DeclareMathSymbol{\sm}{\mathbin}{AMSa}{"39}
\newcommand{\nil}[0]{\mathit{null}}
\newcolumntype{?}{!{\vrule width 1.5pt}}

\newcommand{\Sched}[0]{\mathit{Sc}}
\newcommand{\Structs}[0]{\sigma}

\newcommand{\Env}[0]{\xi}
\newcommand{\Labels}[0]{\mathcal{L}}

\newcommand{\Id}{\mathit{ID}}

\newcommand{\defaultVal}{\mathit{defaultVal}}

\newcommand{\Literals}[0]{\mathit{LT}}

\newcommand{\Program}[0]{\mathcal{P}}

\newcommand{\StructType}[0]{\mathit{sL}}
\newcommand{\SynTypes}[0]{\mathcal{T}}

\newcommand{\interp}[1]{\llbracket #1 \rrbracket}
\newcommand{\true}[0]{\mathit{true}}
\newcommand{\false}[0]{\mathit{false}}
\newcommand{\Stab}[0]{\mathit{s\chi}}
\newcommand{\Stack}[0]{\chi}
\newcommand{\readg}[1]{\textbf{rd}(#1)}
\newcommand{\writev}[1]{\textbf{wr}(#1)}
\newcommand{\cons}[1]{\textbf{cons}(#1)}
\newcommand{\push}[1]{\textbf{push}(#1)}

\newcommand{\Operator}[0]{\textbf{op}}
\newcommand{\Ifc}[1]{\textbf{if}(#1)}
\newcommand{\Notc}[0]{\textbf{not}}
\newcommand{\this}[0]{\textbf{this}}

\newcommand{\ComList}[0]{\gamma}

\newcommand{\Par}[1]{\mathit{Par}_{#1}}

\newcommand{\val}[1]{\mathit{val}(#1)}
\newcommand{\impl}{P_{\mathit{TS}}}

\def\CC{{C\nolinebreak[4]\hspace{-.05em}\raisebox{.4ex}{\tiny\bf ++}}}

\begin{document}
\title{\Lname is Turing Complete}
%
%
\author{Tom T.P. Franken, Thomas Neele}
\date{April 19, 2024}
%
%
\maketitle              
\begin{abstract}
	\Lname is a recently introduced programming language that follows the new data autonomous paradigm.
	In this paradigm, small pieces of data execute functions autonomously.
	Considering the paradigm and the design choices of \Lname, it is interesting to determine the expressiveness of the language and to create verification methods for it.
	In this paper, we take our first steps to such a verification method by implementing Turing machines in \Lname and proving that implementation correct. This also proves that \Lname is Turing complete.
\end{abstract}
\section{Introduction}
Nowadays, performance gains are increasingly obtained through parallelism.
The focus is often on how to get the hardware to process the program efficiently and languages are often designed around that, focusing on threads and processes. 
Recently, \Lname~\cite{franken-autonomous-2023} was introduced, which completely abstracts away from threads.
In \Lname, data is \emph{autonomous}, meaning that the data executes its own functions.
It follows the new data autonomous paradigm~\cite{franken-autonomous-2023}, which abstracts away from active processor and memory management for parallel programming and instead focuses on the innate parallelism of data. 
This paradigm encourages parallelism by making running code in parallel the default setting, instead of requiring functions to be explicitly called in parallel. 
The paradigm also promotes separation of concerns and a bottom-up design process.
A compiler for \Lname~\cite{leemrijse2023} enables execution of \Lname on GPUs.

\Lname is built to be simple and focusses fully on parallel data elements.
This design principle relates \Lname to domain specific languages, which are often less expressive than general purpose languages.
It is therefore relevant to establish the expressiveness of \Lname, as \Lname is built as a general purpose language.
Additionally, establishing the expressiveness of \Lname also indicates how expressive the data-autonomous paradigm is.
\Lname has a fully defined semantics, unlike many other languages, which we can use to answer this question.

\emph{Turing completeness} is a well known property in computer science, which applies to a language or system that can simulate Turing machines.
As a Turing machine can compute all effectively computable functions following the Church-Turing thesis~\cite{copeland-church-turing-1997}, a Turing complete language or system can do the same.
Two approaches to showing Turing completeness are implementing a Turing machine in the target language~\cite{churchill-magic-2020,pitt-turing-2023} and implementing $\mu$-recursive functions~\cite{date-neuromorphic-2022,henderson-turing-2021}.

To prove \Lname's expressiveness, we prove \Lname Turing complete. We do this by implementing a Turing machine in \Lname (Section~\ref{sec: TMADL}). We then give the intuition of the proof that this implementation is correct (Section~\ref{sec: exe}). Constructing this implementation to exhibit correct behaviour is intricate due to \Lname's view on the behaviour of data elements and proofs (specifically those in Appendix~\ref{appendix}) involve detailed reasoning about the semantics and the inference rules defined in it and lay the foundation for proving \Lname programs correct.

\paragraph{Related Work.}
\Lname is a \emph{data-autonomous} language and related to other data-focussed languages, like standard data-parallel languages (CUDA~\cite{garland-parallel-2008} and OpenCL~\cite{chong-sound-2014}), languages which apply local parallel operations on data structures (Halide~\cite{ragan-kelley-halide-2017}, \textsc{ReLaCS}~\cite{raimbault-relacs-1993}) and actor-based languages (Ly~\cite{ungar-harnessing-2010}, A-NETL~\cite{baba--netl-1995}).

Though the expressivity of actor languages has been studied before~\cite{de-boer-decidability-2012} and there is research into suitable Turing machine-like models for concurrency~\cite{qu-parallel-2017, kozen-parallelism-1976, wiedermann-parallel-1984}, there does not seem to be a large focus on proving Turing completeness of parallel languages.
We estimate that this is because many of these languages extend other languages, e.g., CUDA and OpenCL are built upon \CC. 
For these languages, Turing completeness is inherited from their base language.
Furthermore, parallel domain specific languages such as Halide~\cite{ragan-kelley-halide-2017} are simple by design, only focussing on their domain.
Languages may also not be Turing complete on purpose~\cite{gibbons-functional-2015, deursen-little-1998}, for example to make automated verification decidable.

The proof for the Turing completeness of Circal~\cite{detrey-constructive-2002} follows the same line of our proof.
Other parallel systems that have been proven Turing complete include water systems~\cite{henderson-turing-2021} and asynchronous non-camouflage cellular automata~\cite{yamashita-turing-2017}.
	
\section{The Turing Machine Implementation}
\label{sec:preliminaries}
\subsection{Basic Concepts}\label{sec: TM}
We define a Turing machine following the definition of Hopcroft \etal~\cite{hopcroft-introduction-2001}.
Let $\mathbb{D} = \{L, R\}$ be the set of the two directions \textit{left} and \textit{right}. A Turing machine $T$ is a 7-tuple $T = (Q, q_0, F, \Gamma, \Sigma, B, \delta)$, with a finite set of control states $Q$, an initial state $q_0\in Q$, a set of accepting states $F\subseteq Q$, a set of tape symbols $\Gamma$, a finite set of input symbols $\Sigma\subseteq \Gamma$, a blank symbol $B\in \Gamma\setminus\Sigma$ (the initial symbol of all cells not initialized) and a partial transition function $\delta: (Q\setminus F) \times \Gamma \nrightarrow Q \times \Gamma \times \mathbb{D}$.

Every Turing machine $T$ operates on an infinite \textit{tape} divided into \textit{cells}. 
Initially, this tape contains an input string $S = s_0\ldots s_n$ with symbols from $\Sigma$, but is otherwise blank.
The cell the Turing machine operates on is called the \emph{head}. 
We represent the tape as a function $t:\mathbb{Z}\rarr\Gamma$, where cell $i$ contains symbol $t(i)\in\Gamma$. 
In this function, cell $0$ is the head, cells $i$ s.t. $i<0$ are the cells left from the head and cells $i$ s.t. $i>0$ are the cells right from the head.
We restrict ourselves to deterministic Turing machines.	
We also assume the input string is not empty, without loss of generality. 

We define a \emph{configuration} to be a tuple $(q, t)$, with $q$ the current state of the Turing machine and $t$ the current tape function.
Given input string $S = s_0\ldots s_n$, the \emph{initial configuration} of a Turing machine $T$ is $(q_0, t_S)$, with $q_0$ as defined for $T$, and $t_S(i) = s_i$ for $0\leq i \leq n$ and $t_S(i) = B$ otherwise.

During the execution, a Turing machine $T$ performs \emph{transitions}, defined as: 
\begin{definition}[Turing machine transition]\label{def: conf}
	Let $T = (Q, q_0, F, \Gamma, \Sigma, B, \delta)$ be a Turing machine and let $(q, t)$ be a configuration such that $\delta(q, t(0)) = (q', s', D)$, with $D\in \mathbb{D}$. Then $(q, t)\rarr (q', t')$, where $t'$ is defined as
	\begin{equation*}
			t'(i) = 
			\left\{\begin{array}{ll}
				s'&\text{if } i = 1\\
				t(i-1)\;&\text{otherwise}
			\end{array}\right.\text{ if $D{=}L$ and }\\
			t'(i) = 
			\left\{\begin{array}{ll}
				s'&\text{if } i = -1\\
				t(i+1)\;&\text{otherwise}
			\end{array}\right.\text{ if $D{=}R$.}
	\end{equation*}
\end{definition}

We say a Turing machine $T$ \emph{accepts} a string $S$ iff, starting from $(q_0, t_S)$ and taking transitions while possible, $T$ halts in a configuration $(q, t)$ s.t. $q\in F$.

\subsection{The Implementation of a Turing Machine in \Lname}\label{sec: TMADL}
In this section, we describe the implementation of a Turing machine in \Lname.
Let $T = (Q, \Sigma, \Gamma, \delta, q_0, B, F)$ be a Turing machine and $S$ an input string.
We implement $T$ and initialize the tape to $S$ in \Lname.
W.l.o.g., we assume that $Q\subseteq \mathbb{Z}$ with $q_0 = 0$ and that $\Gamma \subseteq \mathbb{Z}$ with $B = 0$.

An \Lname program contains three parts: the definitions of the data types and their parameters are expressed as \emph{structs}, functions to be executed in parallel are given to these data types as \emph{steps}, and these steps are ordered into the execution of a method by a \emph{schedule} separate from the data system.
\textit{Steps} cannot include loops, which are instead managed by the schedule.

We model a cell of $T$'s tape by a struct \textit{TapeCell}, with a left cell (parameter \textit{left}), a right cell (\textit{right}) and a cell symbol (\textit{symbol}).
The control of $T$ is modeled by a struct \textit{Control}, which saves a tape head (variable \textit{head}), a state $q\in Q$ (\textit{state}) and whether $q\in F$ (\textit{accepting}).
See Listing~\ref{ex:sched}.
\begin{lstlisting}[caption={The \Lname program structure}, label={ex:sched}]
struct (*\textit{TapeCell}*) ((*\textit{left}*): (*\textit{TapeCell}*), (*\textit{right}*): (*\textit{TapeCell}*), (*\textit{symbol}*): Int){} //def. of TapeCell
struct (*\textit{Control}*) ((*\textit{head}*): (*\textit{TapeCell}*), (*\textit{state}*): Int, (*\textit{accepting}*): Bool) {
	(*\textit{transition}*) {(*\color{gray} see Listing~\ref{ex:clause} and \ref{ex:transition}*)}      //definition of the step "transition"
	(*\textit{init}*) {(*\color{gray} see Listing~\ref{ex:init}*)}                         //definition of the step "init"
}
(*\textit{init}*) < Fix((*\textit{transition}*))       //schedule: run "init" once and then iterate "transition"\end{lstlisting}
The step \textit{transition} in the \textit{Control} struct models the transition function $\delta$.
For every pair $(q, s)\in Q\times \Gamma$ s.t. $\delta(q, s) = (q', s', D)$ with $D\in \mathbb{D}$, \textit{transition} contains a clause as shown in Listing~\ref{ex:clause} (assuming $D = R$). This clause updates the state and symbol, and saves whether the new state is accepting. It also moves the head and creates a new \textit{TapeCell} if there is no next element, which we check in line 5. For this, as $s$ can be $\nil$, we need to explicitly check whether \textit{head} is a $\nil$-element. Note that $B = 0$, and that if $D = L$ the code only minimally changes.
\begin{lstlisting}[float=t, caption={A clause for $\delta(q, s) = (q', s', R)$.}, label={ex:clause}]
if ((*\textit{state}*) == (*$q$*) && (*\textit{head.symbol}*) == (*$s$*)) then {
	(*\textit{head.symbol}*) := $s'$; //update the head symbol
	(*\textit{state}*) := $q'$; //update the state
	(*\textit{accepting}*) := $(q'\in F)$; //the new state is accepting or rejecting
	if ((*\textit{head}*) != null && (*\textit{head.right}*) == null) then {
		(*\textit{head.right}*) := (*\textit{TapeCell}*)((*\textit{head}*), null, 0); //call constructor to create a new TapeCell
	}
	(*\textit{head}*) := (*\textit{head.right}*); //move right
}\end{lstlisting}

The clauses for the transitions are combined using an if-else if structure (syntactic sugar for a combination of ifs and variables), so only one clause is executed each time \textit{transition} is executed.
See Listing~\ref{ex:transition}.
\begin{lstlisting}[float = t, caption={The \textit{transition} step. The shown pairs all have an output in $\delta$.}, label={ex:transition}]
(*\textit{transition}*) {
	if ((*\textit{state}*) == (*$q_1$*) && (*\textit{head.symbol}*) == (*$s_1$*)) then{ /*clause 1*/ }
	else if ((*\textit{state}*) == (*$q_2$*) && (*\textit{head.symbol}*) == (*$s_2$*)) then { /*clause 2*/ }
	else if ((*\textit{state}*) == (*$q_3$*) && (*\textit{head.symbol}*) == (*$s_3$*)) then { /*clause 3*/ }
	// etc.
} 	\end{lstlisting}
\begin{lstlisting}[float=t, caption={Initializing input string $S$.}, label={ex:init}]
(*\textit{init}*) {
	(*\textit{TapeCell}*) cell0 := (*\textit{TapeCell}*)(null, null, $s_0$); // initialize the tape
	(*\textit{TapeCell}*) cell1 := (*\textit{TapeCell}*)(null, null, $s_1$);
	(*\textit{TapeCell}*) cell2 := (*\textit{TapeCell}*)(null, null, $s_2$);
	cell1.(*\textit{left}*) := cell0; // connect the tape
	cell0.(*\textit{right}*) := cell1;
	cell2.(*\textit{left}*) := cell1;
	cell1.(*\textit{right}*) := cell2;
	(*\textit{Control}*)(cell0, 0, $(q_0\in F)$); //initialize the control
}\end{lstlisting}
In the step \textit{init} in the \textit{Control} struct, we create a \textit{TapeCell} for every symbol $s\in S$ from left to right, which are linked together to create the tape.
We also create a \textit{Control}-instance.
Listing~\ref{ex:init} shows this for an example tape $S = s_0,s_1,s_2$.

In the semantics of \Lname~\cite{franken-autonomous-2023}, the initial state of any program contains only the special $\nil$-element of each struct.
All parameters of the $\nil$-element are fixed to a $\nil$-value.
They can create other elements but cannot write to their own parameters.
Therefore, the call of \textit{init} in the schedule causes the $\nil$-element of \textit{Control} to initialize the tape.
It also initializes a single non-$\nil$ element of \textit{Control}.
The schedule will then have that element of \textit{Control} run the \textit{transition} step until the program stabilizes.
Listing~\ref{ex:sched} shows the final structure of the program.

\section{Turing Completeness}\label{sec: exe}
In this section, we show why \Lname is Turing Complete. 
We establish an equivalence between the configurations of a Turing machine and the configurations that can be extracted from the semantics of the corresponding \Lname program.
We use the fact that the steps executed by the implementation are deterministic, as there is at most one non-$\nil$ \textit{Control} structure that executes the steps.
We omit the full proof of Lemmas~\ref{contract: init} and \ref{contract: transition}, which can be found in Appendix~\ref{appendix}.

Henceforth, let $\impl$ be the implementation of a Turing machine $T$ and an input string $S= s_0\ldots s_n$ as specified in Section~\ref{sec: TMADL}.
In \Lname's semantics, a \emph{struct instance} is a data element instantiated from a struct during runtime.
For the proof we consider a specific kind of \Lname state, the \emph{idle state}, which has the property that none of its the struct instances are in the process of executing a step. In \Lname, the next step to be executed from an idle state is determined by the schedule.
With this we define \emph{implementation configurations}:
\begin{definition}[Implementation Configuration]\label{def: implconf}
	Let $P$ be an idle state of $\impl$ containing a single non-$\nil$ instance $c$ of \textit{Control}.
	Then we define the \emph{implementation configuration} of $P$ as a tuple $(q_P, t_P)$ s.t. $q_P$ is the value of the \textit{state} parameter of $c$ and $t_P: \mathbb{Z}\rarr \mathbb{Z}$ defined as:
	\begin{equation*}
		t_P(i) = \begin{cases}
			c.\mathit{head}.\mathit{symbol} &\text{if } i = 0\\
			c.\mathit{head}.\mathit{left}^{\text{-}i}.\mathit{symbol} &\text{if } i < 0\\
			c.\mathit{head}.\mathit{right}^i.\mathit{symbol} &\text{if } i > 0
		\end{cases},
	\end{equation*}
	where the dot notation $x.p$ indicates the value of parameter $p$ in $x$ and, for $i \geq 1$, $x.p^i$ is inductively defined as $x.p.p^{i-1}$ (with $x.p^0 = x$).
\end{definition}
Note that an implementation configuration is also a Turing machine configuration.
Next we define determinism for \Lname, as well as \emph{data races}.
\begin{definition}[Determinism]
	Let $s$ be a \textit{step} in an \Lname program.
	Then $s$ is deterministic iff for all states that can execute $s$, there exists exactly one state that is reached by executing $s$.
\end{definition}
\begin{definition}[Data Race]
	Let $s$ be a step of $\impl$. Let $P$ be an idle state. Then $s$ contains a data race starting in $P$ iff $P$ can execute $s$ (according to its schedule) and during this execution, there exist a parameter $v$ which is accessed by two distinct struct instances $a$ and $b$, with one of these accesses writing to $v$. We call a data race between writes a \emph{write-write} data race, and a data race between a read a \emph{read-write} data race.
\end{definition}
We use this to prove the following lemma:

\begin{lemma}[\Lname Determinism]\label{lem:det}
	An \Lname step $s$ is deterministic if it cannot be executed by an idle state $P$ in the execution of $\impl$ s.t. $s$ contains a data race starting in $P$.
\end{lemma}
\begin{proof}
	If $s$ contains no data races but is not deterministic, then some parameter $v$ can have multiple possible values after executing $s$ from some idle state $P$. As the operational semantics of \Lname do not allow interleaving by a single struct instance (as defined in the semantics of \Lname~\cite{franken-autonomous-2023}), $v$ must have been accessed by multiple struct instances during execution. The semantics also do not allow randomness, which means that all non-determinism in \Lname results from data races. These struct instances must then be in a data race. This is a contradiction.
\end{proof}
In practice, when a step is deterministic we can ignore interleaving of struct instances during the execution of the step.
\begin{lemma}
	\label{lem: ndinit}
	The execution of \textit{init} in $\impl$ is deterministic.
\end{lemma}
\begin{proof}
	To prove this we need to prove that the execution of \textit{init} contains no data races (Lemma~\ref{lem:det}).
	The step \emph{init} is only executed once, at the start of the program, by the $\nil$-instance of \textit{Control} (as no other instances exist). As only one instance exists, there cannot be a data race between two struct instances.
\end{proof}
\begin{restatable}[Executing \textit{init} in the initial state]{lemma}{initcontract}
	\label{contract: init}
	Let $P_0$ be the idle state at the start of executing $\impl$ and let the input string $S= s_0\ldots s_n$. Executing the step \textit{init} on $P_0$ results in a state $P_1$ with a single non-$\nil$ \textit{Control} instance such that $(q_0, t_S)$ is the implementation configuration of $P_1$.
\end{restatable}
\begin{proof}
	The proof consists of sequentially walking through the statements of \emph{init} when executed from the initial state (which is idle) of $\impl$ as defined in the semantics of \Lname, processing the statements using those semantics.
\end{proof}

\begin{lemma}
	\label{lem: ndtransition}
	Let $P$ be an idle state reachable in $\impl$ with a single non-$\nil$ \textit{Control} instance.
	Any execution of \textit{transition} executed from $P$ is deterministic.
\end{lemma}
\begin{proof}
	As per Lemma~\ref{lem:det}, we prove that the execution contains no data races.
	Let $c$ be an arbitrary clause in the \textit{transition} step (Listing~\ref{ex:clause}).
	If \textit{transition} has a data race during the execution of $c$, this data race must occur between the one non-$\nil$ instance and the $\nil$-instance of \textit{Control}. Let the non-$\nil$ instance be $x_0$ and let $x_1$ be the $\nil$-instance of \textit{Control}.
	Then the parameter which is accessed must be shared by both. This can only be \textit{head.symbol}, as $x_0$ will not get through the if-statement and the other parameters are relative to $x_0$ and $x_1$. However, as $x_0.\mathit{head} = \nil$, this means $\mathit{head.symbol}$ cannot be written to, as parameters of $\nil$-instances cannot be written to in \Lname. This contradicts that it can be in a data race.
\end{proof}
\begin{restatable}{lemma}{impldeterministic}
	\label{lem: dettrans}
	Every \textit{transition} step executed in $\impl$ is deterministic.	
\end{restatable}
\begin{proof}
	By induction. As a base case, the first execution of \textit{transition} happens from $P_1$ as defined in Lemma~\ref{contract: init}, which has only one non-$\nil$ \textit{Control} instance. 
	
	Then consider the execution of \textit{transition} from an idle state $P'$ with one non-$\nil$ \textit{Control} instance, resulting in idle state $P$. Due to Lemma~\ref{lem: ndtransition}, we know that the execution of \textit{transition} is deterministic, so we can consider the sequential execution of \textit{transition}. As \textit{transition} is made up of multiple mutually exclusive clauses, considering only a single clause suffices. 
	As in none of the statements a \textit{Control} instance is created, as seen in Listing~\ref{ex:clause}, it follows that $P$ will also have only a single non-$\nil$ \textit{Control} instance.
\end{proof}
\begin{restatable}[Effect of a \textit{transition} execution]{lemma}{transitioncontract}
	\label{contract: transition}
	Let $P$ be an idle state of $\impl$ from which \textit{transition} can be executed and let $(q,t)$ be the implementation configuration of $P$.
	Assume that $(q, t)$ is also a configuration of $T$.
	Then the result of a transition in $T$ is a configuration $(q', t')$ iff the result of executing the \textit{transition} step from $P$ in $\impl$ is an idle state $P'$ such that $(q',t')$ is its implementation configuration.
\end{restatable}
\begin{proof}
	We know from Lemma~\ref{lem: dettrans} that $p$ has one non-$\nil$ \textit{Control} instance. The proof consists of walking through the statements of \textit{transition} starting at $p$.
\end{proof}	
By induction, using Lemma~\ref{contract: init} as base case and Lemma~\ref{contract: transition} as step, any idle state of $\impl$ after executing \textit{init} corresponds directly to a state $(q, t)$ of $T$, including terminating and accepting states. We conclude:
\begin{restatable}{theorem}{Turingcomplete}
\label{thm: tc}
	\Lname is Turing complete.
\end{restatable}

\section{Conclusion}\label{sec: Con}
In this paper, we have proven \Lname Turing complete by implementing a sequential Turing machine.
In future work, we hope to extend the principles here to a full system to prove \Lname programs correct.
We may also extend the proofs to the weak memory model variant of the \Lname semantics~\cite{leemrijse2023}.
\bibliographystyle{splncs04}
\bibliography{bibliography}

\begin{thebibliography}{10}
\providecommand{\url}[1]{\texttt{#1}}
\providecommand{\urlprefix}{URL }
\providecommand{\doi}[1]{https://doi.org/#1}

\bibitem{baba--netl-1995}
Baba, T., Yoshinaga, T.: A-{NETL}: a language for massively parallel
  object-oriented computing. In: PMMPC Proc. pp. 98--105. IEEE (1995).
  \doi{10.1109/PMMPC.1995.504346}

\bibitem{de-boer-decidability-2012}
de~Boer, F.S., et~al.: Decidability {Problems} for {Actor} {Systems}. In:
  {CONCUR} 2012 – {Concurrency} {Theory}. pp. 562--577. Springer (2012).
  \doi{10.1007/978-3-642-32940-1_39}

\bibitem{chong-sound-2014}
Chong, N., Donaldson, A.F., Ketema, J.: A sound and complete abstraction for
  reasoning about parallel prefix sums. SIGPLAN Not.  \textbf{49}(1),
  397–409 (2014). \doi{10.1145/2578855.2535882}

\bibitem{churchill-magic-2020}
Churchill, A., Biderman, S., Herrick, A.: Magic: {The} {Gathering} {Is}
  {Turing} {Complete}. In: 10th {International} {Conference} on {Fun} with
  {Algorithms} ({FUN} 2021). Leibniz {International} {Proceedings} in
  {Informatics} ({LIPIcs}), vol.~157, pp. 9:1--9:19. Schloss
  Dagstuhl–Leibniz-Zentrum für Informatik, Dagstuhl, Germany (2020).
  \doi{10.4230/LIPIcs.FUN.2021.9}

\bibitem{copeland-church-turing-1997}
Copeland, B.J.: The {Church}-{Turing} {Thesis} (1997),
  \url{https://plato.stanford.edu/ENTRIES/church-turing/}, last Modified:
  2017-11-10

\bibitem{date-neuromorphic-2022}
Date, P., Potok, T., Schuman, C., Kay, B.: Neuromorphic {Computing} is
  {Turing}-{Complete}. In: Proceedings of the {International} {Conference} on
  {Neuromorphic} {Systems} 2022. pp. 1--10. {ICONS} '22, Association for
  Computing Machinery (2022). \doi{10.1145/3546790.3546806}

\bibitem{detrey-constructive-2002}
Detrey, J., Diessel, O.: A Constructive Proof of the Turing Completeness of
  Circal. School of Computer Science and Engineering, University of New South
  Wales, Australia (2002)

\bibitem{deursen-little-1998}
Deursen, A.V., Klint, P.: Little languages: little maintenance? Journal of
  Software Maintenance: Research and Practice  \textbf{10},  75--92 (1998).
  \doi{10.1002/(SICI)1096-908X(199803/04)10:2<75::AID-SMR168>3.0.CO;2-5}

\bibitem{franken-autonomous-2023}
Franken, T.T.P., Neele, T., Groote, J.F.: An {Autonomous} {Data} {Language}.
  In: Theoretical {Aspects} of {Computing} – {ICTAC} 2023. {LNCS}, vol.
  14446, pp. 158--177. Springer International Publishing (2023)

\bibitem{garland-parallel-2008}
Garland, M., et~al.: Parallel {Computing} {Experiences} with {CUDA}. IEEE Micro
   \textbf{28}(4),  13--27 (2008). \doi{10.1109/MM.2008.57}

\bibitem{gibbons-functional-2015}
Gibbons, J.: Functional {Programming} for {Domain}-{Specific} {Languages}. In:
  {CEFP} 2013, pp. 1--28. {LNCS}, Springer International Publishing (2015).
  \doi{10.1007/978-3-319-15940-9_1}

\bibitem{henderson-turing-2021}
Henderson, A., Nicolescu, R., Dinneen, M.J., Chan, T.N., Happe, H., Hinze, T.:
  Turing completeness of water computing. J Membr Comput pp. 182--193 (2021).
  \doi{10.1007/s41965-021-00081-3}

\bibitem{hopcroft-introduction-2001}
Hopcroft, J.E., Motwani, R., Ullman, J.D.: Introduction to automata theory,
  languages, and computation. Addison-Wesley, Boston, 2nd ed edn. (2001)

\bibitem{kozen-parallelism-1976}
Kozen, D.: On parallelism in turing machines. In: 17th {Annual} {Symposium} on
  {Foundations} of {Computer} {Science} (sfcs 1976). pp. 89--97 (1976).
  \doi{10.1109/SFCS.1976.20}

\bibitem{leemrijse2023}
Leemrijse, G.: {Towards relaxed memory semantics for the Autonomous Data
  Language} (2023), {MSc.} thesis, Eindhoven University of Technology

\bibitem{pitt-turing-2023}
Pitt, L.: Turing {Tumble} is {Turing}-{Complete}. Theoretical Computer Science
  \textbf{948},  113734 (2023). \doi{10.1016/j.tcs.2023.113734}

\bibitem{qu-parallel-2017}
Qu, P., Yan, J., Zhang, Y.H., Gao, G.R.: Parallel {Turing} {Machine}, a
  {Proposal}. J. Comput. Sci. Technol.  \textbf{32},  269--285 (2017).
  \doi{10.1007/s11390-017-1721-3}

\bibitem{ragan-kelley-halide-2017}
Ragan-Kelley, J., et~al.: Halide: decoupling algorithms from schedules for
  high-performance image processing. Commun. ACM  \textbf{61},  106--115
  (2017). \doi{10.1145/3150211}

\bibitem{raimbault-relacs-1993}
Raimbault, F., Lavenier, D.: {RELACS} for systolic programming. In: ASAP Proc.
  pp. 132--135. IEEE (1993). \doi{10.1109/ASAP.1993.397128}

\bibitem{ungar-harnessing-2010}
Ungar, D., Adams, S.S.: Harnessing emergence for manycore programming: early
  experience integrating ensembles, adverbs, and object-based inheritance. In:
  OOPSLA Proc. pp. 19--26. ACM (2010). \doi{10.1145/1869542.1869546}

\bibitem{wiedermann-parallel-1984}
Wiedermann, J.: Parallel turing machines. Department of Computer Science,
  University of Utrecht The Netherlands (1984)

\bibitem{yamashita-turing-2017}
Yamashita, T., et~al.: Turing-{Completeness} of {Asynchronous} {Non}-camouflage
  {Cellular} {Automata}. In: Cellular {Automata} and {Discrete} {Complex}
  {Systems}. pp. 187--199. {LNCS}, Springer International Publishing (2017).
  \doi{10.1007/978-3-319-58631-1_15}

\end{thebibliography}

\appendix
\section{Proofs}\label{appendix}
In this appendix, we first present some auxiliary lemmas in Section~\ref{sec:aux}, after which we present the in detail proofs to support Lemmas~\ref{contract: init} and \ref{contract: transition} in Section~\ref{sec:proof}.

The auxiliary lemmas in Section~\ref{sec:aux} use the \Lname semantics to prove that update statements actually perform updates, constructor statements actually construct new data elements and so forth. We recommend reading it only to those familiar with the semantics of \Lname. 
Section~\ref{sec:proof} does not use any concepts other than those introduced in the paper and references to lemmas from Section~\ref{sec:aux} and should not prove a challenge for those who read this paper.
\subsection{Auxiliary \Lname lemmas}\label{sec:aux}
The lemmas presented here are generally useful for any program to be proven correct. They make use of the \Lname semantics, so familiarity with these semantics~\cite{franken-autonomous-2023} is assumed. 

For the rest of this section, let $\Program$ be an \Lname program without read-write data races. Let $P = \langle \Sched, \Structs, \Stab\rangle$ be a state of $\Program$, and let $\ell_p \in\Labels$ be a label and $p$ a struct instance s.t. $p = \Structs(\ell_p) = \langle \StructType_p, \interp{E};\ComList_p, \Stack_p, \Env_p\rangle$. In the first lemma, we prove the result of resolving references in \Lname:

\begin{lemma}[\Lname reference resolution]\label{lem: refres}
	Let $E$ be a variable expression of the form ``\emph{\texttt{$A$.$a$}}'' in \Lname, where $A$ has the syntax ``\emph{\texttt{$a_1$.$\cdots$.$a_n$}}'' with variable identifiers $a, a_1, a_2, a_3, \ldots, a_n$.
	Let the state $P' = \langle \Sched, \Structs', \Stab'\rangle$ be a state resulting from the last transition to resolve $A$, s.t. $\Structs(\ell_p) = \langle \StructType_p, \ComList'_p; \Stack'_p, \Env_p\rangle$. 	
	Then $\Stack'_p = \Stack_p;\ell$, where
	\[\ell = \begin{cases}
		\ell_p &\text{if $n = 0$}\\
		p.a_1.\cdots.a_n &\text{otherwise},
	\end{cases}\]
\end{lemma}
\begin{proof}
	First, $E$ is either a variable that has to be read from, or a variable that has to be written to. In both cases, $\interp{A} = \push{\this};\readg{a_1};\ldots;\readg{a_n}$. We prove by induction on $n$ that for any $n$, the value on the stack resulting from the resolution of $A$ is a single label. As induction hypothesis we take that after the reads up till and including $\readg{a_{j}}$, with $\Stack_{p, j}$ being the stack after those reads, $\Stack_{p, j} = \Stack_p;\ell_{j}$, where $\ell_j$ is the only possible label that can be read from the reads so far.
	\begin{itemize}
		\item[$n = 0$:] If $n = 0$, $\interp{A} = \push{\this}$, which is resolved using the derivation rule \textbf{ComPushThis} and results in $\Stack'_p = \Stack_p;\ell_p$. Therefore, the base case holds.
		\item[$n = i$:] By the induction hypothesis, we know that after the $i-1$th read, $\Stack_{p, i-1} = \Stack_p;p.a_1.\cdots.a_{i-1}$. Then due to the form of $E$ and by the well-typedness of the syntax of $P$, we know that $p.a_1.\cdots.a_{i-1}.a_i\in\Labels$. As there are no read-write data races, we know that there is only a single possible label $p.a_1.\cdots.a_{i-1}.a_i$. Then it follows by the transition \textbf{ComRd} that the read of $p.a_1.\cdots.a_{i-1}.a_i$ removes $p.a_1.\cdots.a_{i-1}$ from the stack and adds $p.a_1.\cdots.a_i$ to it. Therefore, $\Stack_{p, i} = \Stack_p;p.a_1.\cdots.a_i$ and the step holds.
	\end{itemize}
	As the induction holds, the lemma holds.
\end{proof}

We then prove the effect of executing an expression:
\begin{lemma}[\Lname expression execution]\label{lem: expressions}
	Let $E$ be an \Lname expression. Let the state $P' = \langle \Sched, \Structs', \Stab'\rangle$ be a state resulting from the transition of the last command of $\interp{E}$, with $p' = \Structs'(\ell_p) = \langle \StructType_p, \ComList'_p; \Stack'_p, \Env'_p\rangle$. 	
	Then there exists a value $v$ s.t. $\Stack'_p = \Stack_p;v$. Moreover:
	\begin{enumerate}
		\item If $E = $ ``\emph{\texttt{\textbf{this}}}'', $v = \ell_p$.
		\item If $E = $ ``\emph{\texttt{\textbf{null}}}'', and $T$ is the type as determined by the context of $E$, $v = \defaultVal(T)$
		\item If $E = \mathit{lit}$, where $\mathit{lit}\in\Literals$, with semantic value $\val{\mathit{lit}}$, then $v = \val{\mathit{lit}}$.
		\item If $E = x_1.\cdots.x_n.x$, with $x_1,\ldots,x_n\in \Id$, then $v = x_1.\cdots.x_n.x$.
		\item If $E = $ ``\emph{\texttt{$\StructType$ ( $E_1, \ldots, E_n$ )}}'', for a struct type $\StructType$ with parameters\\ $\mathit{par}_1,\ldots, \mathit{par}_n$, then $v\in \Labels$ s.t. $\Structs(v) = \bot$ and $\Structs'(v) = \langle \StructType, \vempty, \vempty, \Env\rangle$. Moreover, with $v_1, \ldots v_n$ as the values resulting from resolving $\interp{E_1},\ldots \interp{E_n}$, $\Env = \Env_{\StructType}^0[\mathit{par}_1 \mapsto v_1, \ldots, \mathit{par}_n\mapsto v_n]$.
		\item If $E = $ ``\emph{\texttt{! $E'$}}'', then $v = \neg v'$, with $v'$ as the value resulting from resolving $\interp{E}$.
		\item If $E = $ ``\emph{\texttt{( $E'$ )}}'', then $v = v'$, with $v'$ as the value resulting from resolving $\interp{E}$.
		\item If $E = $ ``\emph{\texttt{$E_1$ o $E_2$}}'', for syntactic operator $o$ with semantic equivalent $\circ$, then $v = v_1 \circ v_2$, with $v_1$ as the value resulting from resolving $\interp{E_1}$ and with $v_2$ as the value resulting from resolving $\interp{E_2}$
		\item If $E$ creates a new struct instance, $\Stab' = \false^{|\Stab|}$.
	\end{enumerate}
\end{lemma}
\begin{proof}
	We prove the lemma by implicit structural induction, with our induction hypothesis being that the lemma holds for any subexpression encountered in cases 5-8, and with as base cases the cases 1-4. Due to our assumption that $P$ has no read-write data races, we can assume that there is only one possible value read from any variable.
	We then prove the cases separately:
	\begin{enumerate}
		\item If $E = $ ``\texttt{\textbf{this}}'', then according to the interpretation function, this is interpreted as the command $\push{\this}$, which is then resolved using the derivation rule \textbf{ComPushThis}, executed by $p$. The result of this derivation rule is that $p$ pushes $\ell_p$ on it's stack, so this base case holds and $v = \ell_p$.
		\item If $E = $ ``\texttt{\textbf{null}}'', and $T$ is the type as determined by the context of $E$, then according to the interpretation function, this is interpreted as the command $\push(\defaultVal(T))$. This is then resolved using the derivation rule \textbf{ComPush} executed by $p$, which pushes $\defaultVal(T)$ on the stack. Therefore, this base case holds and $v = \defaultVal(T)$.
		\item If $E = \mathit{lit}$, where $\mathit{lit}\in\Literals$, with semantic value $\val{\mathit{lit}}$, then according to the interpretation function, this is interpreted as the command $\push{\val{\mathit{lit}}}$. This is resolved using the derivation rule \textbf{ComPush} executed by $p$, which pushes $\val{\mathit{lit}}$ to the stack. Therefore, this base case holds and $v = \val{\mathit{lit}}$.
		\item If $E = x_1.\cdots.x_n.x$, with $x_1,\ldots,x_n\in V$, then according to the interpretation function, this is interpreted as the commands \\$\push{\this};\readg{x_1};\ldots;\readg{x_n};\readg{x}$. Then, let $P_1$ be the state after resolving $\push{\this};\readg{x_1};\ldots;\readg{x_n}$ with struct environment $\Structs_1$ s.t. $\Structs_1(\ell_p)$ has stack $\Stack_1$. From Lemma~\ref{lem: refres} we know that $\Stack_1 = \Stack;\ell_p.x_1.\cdots.x_n$ and that $\ell_p.x_1.\cdots.x_n \in\Labels$. Then, by \textbf{ComRd}, we know that  $v = \ell_p.x_1.\cdots.x_n.x$. Therefore, this base case holds.
		\item If $E = $ ``\emph{\texttt{$\StructType$ ( $E_1, \ldots, E_n$ )}}'', for a struct type $\StructType$ with parameters\\ $\mathit{par}_1,\ldots, \mathit{par}_n$, then from the interpretation function, we know that $E$ is interpreted as $\interp{E_1};\ldots;\interp{E_n};\cons{\StructType}$. By the structural induction hypothesis, we know that $\interp{E_1};\ldots;\interp{E_n}$ results in the sequence of values $v_1;\ldots;v_n$ at the end of the stack of $p$. Then by the derivation rule \textbf{ComCons}, we know that there exists a label $\ell$ s.t. $\Structs(\ell) = \bot$ and $\Structs'(\ell) = \langle \StructType, \vempty, \vempty, \Env\rangle$, where $\Env = \Env_{\StructType}^0[\mathit{par}_1 \mapsto v_1, \ldots, \mathit{par}_n\mapsto v_n]$. Also by \textbf{ComCons}, we know that $v = \ell$. Therefore, this case holds
		\item If $E = $ ``\emph{\texttt{! $E'$}}'', then by the interpretation function, this gets interpreted as $\interp{E'};\Notc$. Then by the structural induction hypothesis, we know that $\interp{E'}$ results in a value $v'$ at the end of the stack of $s$. Then by derivation rule \textbf{ComNot}, we know that $v = \neg v'$, so this case holds.
		\item If $E = $ ``\emph{\texttt{( $E'$ )}}'', then as the concrete syntax gets converted into an abstract syntax tree, $\interp{E} = \interp{E'}$, as $\interp{E'}$ pushes a value $v'$ to the stack as per the structural induction hypothesis, it follows that $\interp{E}$ also pushes $v'$ to the stack, so $v = v'$. Therefore, this case holds.
		\item If $E = $ ``\emph{\texttt{$E_1$ o $E_2$}}'', for syntactic operator $o$ with semantic equivalent $\circ$, this is interpreted by the interpretation function as $\interp{E_1};\interp{E_2};\Operator(o)$. Then by the structural induction hypothesis, we know that $\interp{E_1};\interp{E_2}$ results in the values $v_1;v_2$ on the stack of $p$. Then by derivation rule \textbf{ComOp}, we know that $v = v_1 \circ v_2$. Therefore, this case holds.
		\item If $E$ creates a new struct instance, then either $E = $ ``\emph{\texttt{$\StructType$ ( $E_1, \ldots, E_n$ )}}'' or a subexpression of $E$ creates a new struct instance. In the second case, $\Stab'=\false^{|\Stab|}$ by the structural induction hypothesis. In the first case, due to the execution of \textbf{ComCons} during the resolution of $E$, $\Stab' =\false^{|\Stab|}$.
	\end{enumerate}
	As the structural induction and all cases within it hold, the lemma holds.
\end{proof}

We can use this lemma to prove the effects of statement executions. Firstly, for constructor statements, note the following:
\begin{corollary}[\Lname constructor execution]\label{lem: constructorex}
	A constructor statement has the same effects as a constructor expression, as defined in case 5 of Lemma~\ref{lem: expressions}, and also resets the stability stack as defined in case 9 of Lemma~\ref{lem: expressions}.
\end{corollary}

We then prove the update statement execution effects:
\begin{lemma}[\Lname update execution]\label{lem: updateex}
	Let $Z$ be an \Lname statement of the form ``\emph{\texttt{$A$.$a$ := $E$}}'', where $A$ has the syntax ``\emph{\texttt{$a_1$.$\cdots$.$a_x$}}'' with $a,a_1,a_2,a_3,\ldots,\linebreak[4]a_n\in \Id$ and $E$ is an expression. Let label $\ell_\alpha$ be uniquely defined as
	\[\ell_\alpha = \begin{cases}
		p.a_1.a_2.a_3.\cdots.a_n &\text{if } n > 0\\
		\ell_p &\text{if } n = 0
	\end{cases}.\]
	Let the state $P' = \langle \Sched, \Structs', \Stab'\rangle$ be the state resulting from the transition of the last command of $\interp{Z}$. Let $v$ be the value pushed to the stack as a result of resolving $\interp{E}$ (as per Lemma~\ref{lem: expressions}). Let $\Structs(\ell_\alpha) = \langle \StructType_\alpha, \ComList_\alpha, \Stack_\alpha, \Env_\alpha\rangle$ and let $\Structs'(\ell_\alpha) = \langle \StructType_\alpha, \ComList'_\alpha, \Stack'_\alpha, \Env'_\alpha\rangle$. 
	
	Then, if $\ell_\alpha \neq \Labels_0$, $\Env'_\alpha(a) = v$ and if $\Env_\alpha(a) \neq v \land a\in \Par{\StructType_\alpha}$, $\Stab' = \false^{|\Stab|}$.
\end{lemma}
\begin{proof}
	We know that $\interp{Z} = \interp{E};\interp{a_1;\cdots;a_x};\writev{a}$ by the definition of the interpretation function. From Lemma~\ref{lem: expressions} we know that through $\interp{E}$, $v$ is put on the stack first. Then, by Lemma~\ref{lem: refres}, we know that the result of $\interp{a_1;\cdots;a_x}$ is that $\ell_\alpha$ is pushed on the stack. Then if $\ell_\alpha \neq \Labels_0$, we know through the derivation rule \textbf{ComWr} that $\Env_\alpha(a) = v$. If $\Env_\alpha(a)\neq v \land a\in\Par{\StructType_\alpha}$, we know that the value of $a$ before \textbf{ComWr} can either still be $\Env_\alpha(a)$ or it can have been written to by another struct instance $p'$, also using a \textbf{ComWr} transition. In the first case, it follows from \textbf{ComWr} that $\Stab' = \false^{|\Stab|}$. In the second case, if the other struct instance writes $v$ to $a$, $\Stab' = \false^{|\Stab|}$ due to the \textbf{ComWr} transition done by $p'$, and if not, then $\Stab' = \false^{|\Stab|}$ due to the \textbf{ComWr} transition of $p$. In any case, $\Stab' = \false^{|\Stab|}$.	 
\end{proof}

The above lemma also suffices for assignment statements:
\begin{corollary}[Assignment statements]\label{lem: assignmentex}
	Lemma~\ref{lem: updateex} also holds for statements of the form ``\emph{\texttt{$T$ $a$ := $E$}}'', with $T\in\SynTypes$ and ``\emph{\texttt{$a$ := $E$}}''.
\end{corollary}
\begin{proof}
	As $\interp{\text{``\emph{\texttt{$T$ $a$ := $E$}}''}} = \interp{\text{``\emph{\texttt{$a$ := $E$}}''}}$, the effects of executing ``\emph{\texttt{$T$ $a$ := $E$}}'' are the same as executing ``\emph{\texttt{$a$ := $E$}}'' (with $n = 0$). The stability stack will not be updated, as $a$ cannot be a parameter according to our static syntax requirements.
\end{proof}

Lastly, we prove the effects of executing an if-then statement:
\begin{lemma}[If-then Statements]\label{lem: ifthen}
	Let $Z$ be a statement of the form ``\emph{\texttt{if $E$ then \{ $S$ \}}}'', where $S$ is a list of statements and $E$ is an expression. 
	
	Let the state $P' = \langle \Sched, \Structs', \Stab'\rangle$ be the state resulting from the transition of the last command of $\interp{Z}$, and let $\Structs'(\ell_p) = \StructType_p, \ComList'_p, \Stack'_p, \Env'_p\rangle$. Let $v$ be the value pushed to the stack as a result of resolving $\interp{E}$ (as per Lemma~\ref{lem: expressions}).
	
	Then either $v = \true$ and $\ComList'_p = \interp{S};\ComList_p$ or $v = \false$ and $\ComList'_p = \ComList_p$. 
\end{lemma} 
\begin{proof}
	We know that $\interp{Z} = \interp{E};\Ifc{\interp{S}}$ by the definition of the interpretation function. By our assumption of well-typedness, we know that the value $v$ to which $\interp{E}$ resolves is a boolean value, and therefore the value at the end of the stack after resolving $\interp{E}$ will be either $\true$ or $\false$. Then if $v = \true$, we know by the derivation rule \textbf{ComIfT} that $\ComList'_p = \interp{S};\ComList_p$, and if $v = \false$, we know by the derivation rule \textbf{ComIfF} that $\ComList'_p = \ComList_p$.
\end{proof}
We have now proven the effects of every type of statement. For if-statements and constructor statements, the result of the statement are permanent during the execution of a step $s$. Assignment statements can only work with local variables, of which the values are irrelevant at the end of $s$. We do however need to prove what we can guarantee about updated parameters after the execution of an update:

\begin{lemma}[\Lname update results]\label{lem: update}
	Let $s$ be a step in $\Program$ and let $Z$ be an update statement s.t. $p$ executing $Z$ updates a parameter $p'.x$ with type $T$ of some struct instance $p'$ to a value $b$ during $s$ (along Lemma~\ref{lem: updateex}). Let $a$ be the original value of $p'.x$. Let $P$ be a state during $s$ after the execution of $Z$ by $p$ and before the execution of the statement after $Z$ by $p$. Then all of the following holds:
	\begin{itemize}
		\item[a.] If $p'$ is a $\nil$-instance, $p'.x = a = \defaultVal(T)$.
		\item[b.] If $p'$ is not a $\nil$-instance:
		\begin{itemize}
			\item[i.] If $p'.x$ is not involved in a write-write data race, $p'.x = b$.
			\item[ii.] If $p'.x$ is involved in a write-write data race, let $N$ be the set of all values written to $p'.x$ during $s$ by all data elements involved in the write-write data race. Then $p'.x\in N$.
		\end{itemize}
	\end{itemize} 
	Additionally, we know that if $a \neq b$, the stability stack is reset.
\end{lemma}
\begin{proof}
	If $p'$ is a $\nil$-instance, then by the rule \textbf{ComWrSkip} and by the initialization of $\nil$-instances, we know that $p'.x = a = \nil$ after the execution of $Z$ by $p$.
	If $p'$ is not a $\nil$-instance, and $p'.x$ is not involved in a write-write data race, $p'.x$ is not involved in any data race, as $\Program$ does not have read-write data races. Then as no other element other than $p$ can have written to $p'.x$ during or after the execution of $Z$ by $p$ and $b$ is deterministic during the execution of $Z$ (as there are no read-write data races), $p'.x = b$. If $p'$ is not a $\nil$-instance and $p'.x$ is in a write-write data race, as \Lname does not allow for nondeterminism in a single data element, this data race must be between different data elements. As these can execute their update statements in any order, any value in $N$ can be the last value written to $p'.x$ before $p$ executes the statement after $Z$, so $p'.x\in N$.
\end{proof}

This extends to step executions:
\begin{corollary}[\Lname parameters after a step]\label{cor: steppar}
	Let $s$ be a step in $\Program$ and let $P$ be a state resulting from an execution of $s$. Let $Z$ be the last update statement in $s$ of some parameter $p'.x$ with type $T$ of some struct instance $p'$ by $p$, which updates $p'.x$ to a value $b$. Let the value of $p'.x$ before the execution of $s$ be $a$. Then in $P$:
	\begin{itemize}
		\item[a.] If $p'$ is a $\nil$-instance, $p'.x = a = \defaultVal(T)$.
		\item[b.] If $p'$ is not a $\nil$-instance:
		\begin{itemize}
			\item[i.] If $p'.x$ is not involved in a write-write data race, $p'.x = b$.
			\item[ii.] If $p'.x$ is involved in a write-write data race, let $N$ be the set of all values written to $p'.x$ during $s$ by all data elements involved in the write-write data race. Then $p'.x\in N$.
		\end{itemize}
	\end{itemize} 
	Additionally, we know that if $p'.x$ has had its value changed during the execution of $s$, the stability stack has been reset during the execution of $s$.
\end{corollary}

\subsection{Turing Complete Lemmas Proofs}\label{sec:proof}
In this section, we will prove Lemma~\ref{contract: init} and \ref{contract: transition} in more detail, using the auxiliary lemmas of the previous section.
Recall Lemma~\ref{contract: init}:
\initcontract*
\begin{proof}
	First, note that the idle state at the start of executing $\impl$ is the initial state of $\impl$. The initial state for $\impl$, as defined in the \Lname semantics, contains the schedule of $\impl$, the $\nil$-instances of all structs, and a stability stack. The stability stack has no bearing on this proof, and will be disregarded.
	
	We need to show that $p_1$ has implementation configuration $(q_0, t_S)$. To show that, we first prove that $p_1$ contains only one non-$\nil$ struct instance of \textit{Control} and that its \textit{state} parameter is set to $q_0$.
	To prove this, we can assume the \textit{init}-code is executed without nondeterministic behaviour, due to Lemma~\ref{lem: ndinit}. Only the $\nil$-instance of \textit{Control} executes \textit{init} (as it is the only instance to exist in $p_0$). The step code makes only a single \textit{Control}-instance, and as the code is deterministic and the $\nil$-instance executes it, we know that this means only one \textit{Control}-instance is present in $p_1$, following Lemma~\ref{lem: constructorex}, which we will call $c$. Also following Lemma~\ref{lem: constructorex}, we know that the \textit{state} parameter of $c$ is set to $q_0$ (represented by integer $0$).
	
	We then prove that the function made according to Definition~\ref{def: implconf} in $p_1$ from the \textit{TapeCells} is $t_S$. To do this, we first prove that in $p_1$, there exists a \textit{TapeCell} for all symbols $s_i\in S$, and no others, s.t. every \textit{TapeCell} $s_i$ is be connected to $s_{i-1}$ and $s_{i+1}$ (if they exist) through parameters \textit{left} and \textit{right} respectively. Then we prove that the \textit{head} parameter of $c$ will be set to the \textit{TapeCell} of $s_0$. 
	
	The first follows from the template in Listing~\ref{ex:init}, which we can follow sequentially due to Lemma~\ref{lem: ndinit}. According to Lemma~\ref{lem: assignmentex}, the first part makes one \textit{TapeCell} instance for every $s_i\in S$, and according to Corollary~\ref{lem: update} and Lemma~\ref{lem: ndinit}, these are then connected to the correct \textit{left} and \textit{right} neighbours.
	It follows from the listing and Lemma~\ref{lem: constructorex} that \textit{head} parameter of $c$ will be set to the \textit{TapeCell} for $s_0$.
	
	Then the lemma holds: the implementation configuration of $p_1$ is $(q_0, t_S)$. \end{proof}
	
Now recall Lemma~\ref{contract: transition}:
\transitioncontract*
\begin{proof}
	Let $p$ and $(q, t)$ be as defined in the lemma. Then by definition of $\impl$, there exists a single transition in $p$ for $\impl$ iff $\delta(q, t)$ is defined. Additionally, $\delta$ is a function, so it is always uniquely defined for $(q, t)$. 
	
	Let $\delta(q, t) = (q', s', D)$. W.l.o.g., let $D = R$ (the proof of $D = L$ is analogous). Then as a transition in $T$ is deterministic, by Definition~\ref{def: conf}, the resulting state of taking a transition from $(q, t)$ is the state $(q', t')$, with $$t'(i) = \begin{cases}
		s'&\text{if } i = -1\\
		t(i+1) &\text{otherwise}
	\end{cases}.$$
	
	Taking the transition in $p$ for $\impl$ is also deterministic (Lemma~\ref{lem: dettrans}), and therefore we can walk through the statements of the clause to determine its effect. By definition of $\impl$, this clause is based on the template shown in Listing~\ref{ex:clause}. Let $c$ be the single \textit{Control} instance of $p$, and let $h$ be the \textit{TapeCell} instance which is referenced in the \textit{head} parameter of $c$. Due to Lemma~\ref{lem: update}, the result of the transition is that the \textit{state} of $c$ is updated to $q'$, the \textit{symbol} of $h$ is updated to $s'$, the \textit{accepting} parameter of $c$ is updated to whether $q'\in F$ and that the \textit{head} parameter of $c$ shifts one \textit{TapeCell} to the right (making a new \textit{TapeCell} if required, by Lemma~\ref{lem: ifthen} and Lemma~\ref{lem: update}).
	
	Creating a function of the \textit{TapeCells} as in Definition~\ref{def: implconf} then results in function $t''$, s.t. $t''(-1) = s'$ and $t''(i) = t(i+1)$ for all $i \neq -1$, which is equal to $t'$. Then, by Definition~\ref{def: implconf}, the implementation configuration of the resulting state is $(q', t')$.
	Therefore the lemma holds.
	\end{proof}
\end{document}